\documentclass{ectj}
\usepackage{amsfonts,amssymb,graphics,epsfig,verbatim,bm,latexsym,amsmath,url,amsbsy}
\usepackage{natbib}

\newtheorem{theorem}{Theorem}
\newtheorem{assumption}{Assumption}

\newtheorem{corollary}{Corollary}
\newtheorem{lemma}{Lemma}

\newtheorem{definition}{Definition}

\renewcommand{\thesection}{\arabic{section}}
\renewcommand{\theequation}{\arabic{section}.\arabic{equation}}

\received{...}
\accepted{...}
\volume{21}

\title[Identification of a class of index models]{Identification of a class of index models: A topological approach\thanks{
This project has received funding from the European Research Council (ERC)
under the European Union's Horizon 2020 research and innovation programme
(grant agreement No. 740369).}}

\author{Mogens Fosgerau$^{\dagger}$ and
                Dennis Kristensen$^{\ddagger}$}

\address{$^{\dagger}$Dept. of Economics, Univ. of Copenhagen, Øster Farimagsgade 5,
1353 København K, Denmark.}
\email{mogens.fosgerau@econ.ku.dk}

\address{$^{\ddagger}$Department of Economics, University College London, Gower Street, London, WC1E 6BT, UK.}
\email{d.kristensen@ucl.ac.uk}

\def\AmSTeX{$\cal A$\kern-.1667em\lower.5ex\hbox{$\cal M$}\kern-.125em
    $\cal S$-\TeX}
\def\BibTeX{{\rm B\kern-.05em{\sc i\kern-.025em b}\kern-.08em
    T\kern-.1667em\lower.7ex\hbox{E}\kern-.125emX}}

\begin{document}

\begin{abstract}
We establish nonparametric identification in a class of so-called index models using a novel approach that relies on general topological results. Our proof strategy requires substantially weaker conditions on the functions and distributions characterizing the model compared to existing strategies; in particular, it does not require any large support conditions on the regressors of our model. We apply the general identification result to additive random utility and competing risk models.
\end{abstract}

\keywords{Nonparametric identification, Discrete choice, Competing risks, Index model.}

\section{Introduction}

We develop a novel nonparametric identification result for the following class of models,
\begin{equation}
\Pi \left( w,x,z\right) =\Lambda \left( a(w,x),z\right) ,  \label{eq:model}
\end{equation}where
\begin{equation}
a\left( w,x\right) =g\left( w\right) +h\left( x\right) .
\label{eq: a(w,x) def}
\end{equation}
is a vector of additively separable index functions while $\Lambda :%
\mathbb{R}^{J}\times \mathbb{R}^{d_{Z}}\mapsto \mathbb{R}^{J}$, $g:\mathbb{R}^{d_{W}}\mapsto \mathbb{R%
}^{J}$ and $h:\mathbb{R}^{d_{X}}\mapsto \mathbb{R}^{J}$ are all vector-valued functions of dimension $J\geq 1$. The arguments $w\in \mathbb{R%
}^{d_{J}}$ and $x\in \mathbb{R}^{d_{X}}$ represent the values of two sets  of regressors, $W$ and $X$, while $z\in \mathbb{R}^{d_{Z}}$ corresponds to values of a set of control variables, $Z$. We take as high-level assumption that we know (have observed from data) the function $\Pi \left( w,x,z\right)
$, for $\left( w,x,z\right)$ in the support of $\left( W,X,Z\right) $, from which we then wish to identify the unknown functions $\Lambda \left( a,z\right) $ and $h\left( x\right) $, while we treat the function $g\left( w\right) $ as being known. We refer to this class of models as index models since $W$ and $X$ are restricted to enter the model through $g(W)$ and $h(X)$, respectively. We make three major contributions relative to the existing literature:

First, we do not impose any large support conditions on any of the regressors in our model. Most existing results on identification within this class of models require availability of a set of "special" continuously distributed regressors; identification is then achieved by sending each of these special regressors off to the boundary of their support. Estimators based on such "thin set identification" argument were analyzed by \citet{Khan2010} who showed that they tend to be irregularly behaved with slow convergence rates. In contrast, we achieve identification as long as the random index $a\left( W,X\right)$ exhibits sufficient, but potentially bounded, variation. We expect this to translate into better behaved estimators.

Second, we impose weak conditions on the functions of interest and distributions of the random variables $(W,X,Z)$. We do not require continuity or differentiability of the functions entering the model  in order to show identification while most existing results as a minimum require these to be differentiable. Similarly, we only require $g(W)$ to have continuous support while $(X,Z)$ can both be discrete, continuous or a mix of the two as long as their supports satisfy certain conditions.  Thus, our results cover models with thresholds and kinks in $\Lambda $, $g$ and $h$, which existing results cannot handle. In the case of discrete choice models such features may occur if the decision maker optimizes subject to constraints; see, e.g., \citet{Cantillo2006a}. These models have traditionally been formulated in a parametric fashion; our theory demonstrates how these can be identified without parametric constraints. There is a growing literature on nonparametric estimation with unknown thresholds and kinks which we conjecture can be employed in our setting in order to translate our identification result into actual estimators; see, e.g., \citet{Chiou2018}.

Third, we show how the presence of the controls $Z$ can help to achieve identification in a nontrivial way: We first show local identification at each value of the control $Z$. Suitable variation in $Z$ then allows us to piece the locally identified components together across different values of $%
Z$ to achieve global identification. In comparison, most other papers that allow for control variables show identification at a fixed arbitrary value of $Z$ in which case variation in $Z$ is unnecessary for identification.

Our proof strategy relies on arguments from general topology that, to our knowledge, are completely new to the literature on nonparametric identification. These should be of general interest since they can be used for identification in other settings. The two key elements of our approach is the notions of relative identification and connected sets. Below, we state our formal definition of the former:

\begin{definition}
A function $h$ is said to be \textit{relatively identified} on a given set $\mathcal{X}$ if identification of $h \left(x^{*} \right)$ at some point $x^{*}\in \mathcal{X}$ implies that $h \left( x \right)$ is also identified at all other $x \in \mathcal{X}$.
\end{definition}

Next, recall the topological notion of connnectedness: A connected set cannot be contained in the union of two non-empty disjoint open sets while having non-empty intersection with both. In particular, it is not possible to split a connected open set into disjoint open subsets.

Our identification strategy then proceeds in three steps where we here initially suppress the presence of $Z$ for simplicity: First, we decompose the support of $X$ into suitable subsets and achieve relative identification on each of these. This is done via two features of our model: For a given $ x$, we are able to identify the relative variation in $\Lambda \left(a \right)$, with $a = h \left( x \right)+ g\left(w \right)$, through the observed variation in $\Pi\left(w,x\right)$ w.r.t. $w$ through the known function $g(w)$. By injectivity of $%
\Lambda $, we are then able to identify the relative value of $a$ which in turn yields the relative value of $h \left( x \right)=a - g\left(w \right)$ on suitably chosen subsets of the support of $X$. Second, we achieve global identification on the union of these subsets by using the second main ingredient of our proof strategy, connectedness: We will require the support of $%
a\left( W,X\right) $ to be connected which is used to extend relative local identification to global identification. Finally, reintroducing $Z$, we again rely on the supports of $X|Z=z$ to be suitably connected across different values of $z$ in the support of $Z$ to enlarge the identification region further.

Like us, \citet{Berry2018} and \citet{Evdokimov2010}, among others, rely on connectness to achieve global identification but in these papers the restriction is imposed directly on the support of the covariates thereby implicity restricting the covariates to be continuous. In contrast, we impose connectedness on the image of $a\left( W,X\right) $ and so allow for both $X$ and $W$ to contain discrete components.

Two leading examples that fall within our general framework are nonparametric additive versions of multiple discrete choice and competing risk models as shown in the next section. There is a large literature on identification and estimation of semiparametric multinomial choice models \citep[see, e.g.,][]{Manski1975,Lewbel2000}. In contrast, the literature on nonparametric identification is quite thin with few results having been developed since the seminal work of \citet{Matzkin1993}. In terms of modelling, \citet{Matzkin1993} is probably the most closely related to our setting, but the assumptions made and identification strategy pursued in this paper are very different from ours. Our and her set of assumptions are not clearly ranked with some of our assumptions being stronger while others weaker compared to hers. One key feature of her proof strategy is the introduction of assumptions that ensure the multinomial model can be converted into a binary choice problem followed by a thin-set identification argument. More recently, \citet{Allen2019} provide conditions under which one can identify how regressors alter the desirability of alternatives using only average demands. Their conditions are weaker than ours but on the other hand they are only able to identify certain features of the model, not the underlying data-generating structure.

There is also a nascent literature on nonparametric identification of so-called BLP models \citep[][]{Berry1995} as used in industrial
organization; see, for example, \citet{Berry2018} and \citet{Chiappori2018}. The setting of the BLP model is somewhat different, though, since there the choice probabilities are treated as observed variables which depend on unobserved product characteristics that have to be controlled for. This leads to a different identification problem compared to ours.

Finally, there is also a literature on identification in competing risk models. The two most closely related papers in terms of modelling are \citet{Heckman1989} and \citet{Lee2013}. \citet{Heckman1989} achieves identification by assuming the index (in our notation $a(W,X)$) has support  on $(0,\infty)^J$ and then achieves identification of a given component of the index by letting the other components go to zero, and so their result falls in the thin-set identification category. \citet{Abbring2003a} weaken this assumption substantially for the class of mixed proportional hazard models, a subclass of competing risk models.  \citet{Lee2013} provide a high-level assumption for identification of the general model involving a rank condition of an integral operator. Primitive conditions for this to hold are not known. \citet{Honore2006} derive bounds for the functions of interest when only discrete covariates are available. We complement these studies by showing identification in the general competing risk model under primitive conditions that allow for the presence of discrete co-variates, but at the same time impose more structure on the index, c.f. eq. (\ref{eq: a(w,x) def}). 

In the next section, we give two motivating examples in form of a random utility model and a competing risk model that both fall within the setting of eq. (\ref{eq:model}). We present our general framework in Section \ref{Sec: Identification} and the assumptions we will work under, and provide our identification results in Section \ref{sec:results}. Section \ref{Sec:Application} applies our general result to the two examples and Section \ref{Sec:Conclusion} concludes.

\section{Two Motivating Examples}
\label{Sec: Examples}

The model (\ref{eq:model}) comprises a range of models that are met in economics. We here present two classes of models that fall within our framework. We will return to these two classes of models in Section \ref{Sec:Application} where we apply our general identification result to each of them.

\subsection{Discrete choice models}

We here first demonstrate that the class of additive random utility models (ARUM) can be mapped into (\ref{eq:model}). Using existing results in the literature, this in turn implies that our results also apply to a broad class of rational inattention discrete choice models \citep[][]{Fosgerau2016r} and an even wider class of perturbed utility models.

\subsubsection{Additive random utility}

Consider an agent choosing between $J+1$ alternatives, each carrying an associated indirect utility of the form
\begin{equation*}
U_{j}= a_{j}\left( W, X\right)+\varepsilon _{j},\text{ \
\ }j=0,1,...,J, \label{ARUM model}
\end{equation*}where $\left( W,X\right) $ is a set of observed covariates while $\varepsilon =\left( \varepsilon _{0},\varepsilon _{1},...,\varepsilon
_{J}\right) $ is unobserved. This model was initially proposed by \citet{McFadden1974a} and has since become one of the workhorses in applied microeconomics; see e.g. \citet{Ben-Akiva1985} and \citet{Maddala1983}. As is standard in the literature, we impose the following normalization on the "outside"{}option $j=0$: $a_{0}\left( w,x\right) =0$.

Some of the regressors $\left( W,X\right) $ may potentially be dependent on $%
\varepsilon $. To handle this situation, we assume the availability of a set
of control variables $Z$ so that $(W,X)$ are independent of $\varepsilon$ conditional on $Z$. In addition to $\left( W,X,Z\right) $, the researcher also observes the
utility maximizing choice, $D=\arg \max_{j\in \left\{ 0,1,...,J\right\}
}U_{j}$. Thus, the conditional choice probabilities (CCP's),
\begin{equation}
\Pi _{j}\left( w,x,z\right) :=P\left( D=j|\left( W,X,Z\right) =\left(
w,x,z\right) \right) ,\text{ \ \ }j=0,1,...,J,  \label{eq: Pi ARUM}
\end{equation}are identified in the population. We collect these in the vector-valued function $\Pi \left( w,x,z\right) =\left\{ \Pi _{j}\left( w,x,z\right)
:j=1,...,J\right\} \in \mathbb{R}^{J}$ where we leave out the CCP of the outside option. It now follows from standard results in the literature that $\Pi \left( w,x,z\right) $ can be written on the form (\ref{eq:model}) with $\Lambda$  being the  gradient of the so-called surplus function; see Section \ref{Sec:Application} for further details.

Our identification result requires the researcher to group the observed covariates into two sets: The first set, denoted $W$, contains the "special" regressors that enter the index $a$ through a known function $g(W)$ as specified by the researcher, c.f. eq. (\ref{eq: a(w,x) def}). The second set, denoted $X$, then enters $a$ through $h(X)$ which is left unspecified. The choices of $W$ and $g(W)$ are application specific and should be guided by two considerations: First, $g(W)$ need to exhibit sufficient continuous variation on $\mathbb{R}^J$ since this is a key requirement for our identification result to go through. Second, since $g_j(W)$  affects the utility of the $j$th alternative positively by definition, it should be specified accordingly. 

As an example of this joint modelling and identification strategy, let us consider the problem of estimating willingness-to-pay for different goods, a common problem  in various applied fields of economics  \citep[e.g.,][]{Fosgerau2006d,Bontemps2016}. In this setting, choosing $g$ to be $g_j(W_j)=-\ln W_j$, where $W_j$ is the price of alternative $j$, $j=1,...,J$, transforms a positive price vector into a vector that can in principle attain values in all of $\mathbb{R}^J$. With this choice, $h_j(X)+\varepsilon_j$ captures the log willingness to pay for good $j$, where $X$  contains characteristics of the agent and other characteristics of the different alternatives. Prices generally exhibit continuous variation and so satisfy the first of the two aforementioned requirements. This example assumes the availability of alternative specific regressors, $W_1,....,W_J$. However, our identification result may still be applied if this is not true. In this case, the researcher needs to construct alternative-specific regressors $g_1(W),....,g_J(W)$ from a set of underlying covariates $W$. 

Our assumption of $g(W)$ being known has antecedents in the literature on identification in discrete choice models.  For example, in the context of binary choice ($J=1$), \cite{Lewbel2000} also assumes the presence of a "special" regressor, in our notation $W$, that enters the utility of alternative 1 in a known fashion. But this paper furthermore restricts $h(x)$ to be linear, $h(x)=\beta x$  and, importantly, identification of $\beta$ is achieved through variation of $g(W)$ on the boundary of its support. Our identification result does not rely on any such argument.

Our framework also includes so-called rational inattention discrete choice model. \citet{Fosgerau2016r} show that any ARUM satisfying the conditions above is observationally equivalent to a rational inattention discrete choice model in which the prior is held constant. This generalizes the finding of \citet{Matejka2015} who show that the multinomial logit model has a foundation as a rational inattention model. Thus, our identification result extends without effort to a broad class of rational inattention models.

\subsubsection{Perturbed utility}

The class of perturbed utility models \citep[][] {McFadden2012,Fudenberg2015,Allen2019} is another generalization of the class of ARUM. As shown by \citet{Hofbauer2002}, the CCP's of an ARUM can be represented as the solution to a maximization problem where an agent chooses the vector of CCP's to maximize a function that consists of a linear term and a concave term. Here we present an extended version that includes controls affecting the concave term, i.e.
\begin{equation}
\Lambda\left(a,z\right)=\arg\max_{q\in\Delta}\left\{
a^{\intercal}q+\Omega\left(q|z\right)\right\} , \label{pertubed utility}
\end{equation}
where $a\in\mathbb{R}^{J+1}$ is a vector of utility indices, $\Delta=\{q\in%
\mathbb{R}_{+}^{J+1}:\sum_{j=0}^{J}q_{j}=1\}$ is the unit simplex and $%
\Omega\left(\cdot|z\right)$ is a concave function for each $z\in\mathcal{Z}$.  
 The perturbed utility model includes ARUM as a special case, while allowing an individual to have strict preference for randomization rather than to choose a vertex of the probability simplex. As noted by \citet{Allen2019}, observing only realizations of lotteries across choice options is sufficient for identification which requires only the vector of CCP's, $\Pi\left(w,x,z\right)\ $. We show in Section \ref{Sec:Application} that the implied CCP's satisfy (\ref{eq:model}).

\subsection{Accelerated failure time models for competing risks}

Consider a competing risk model as in \citet{Heckman1989} with $J$ competing causes of failure. A latent failure time $T_{j}>0$ is associated with each cause $j\in \left\{ 1,...,J\right\} $. The econometrician observes the duration until the first failure, $Y=\min_{j\in \left\{ 1,...,J\right\}
}T_{j}$, and the associated cause of failure, $D=\arg \min_{j\in \left\{
1,...,J\right\} }T_{j}$, together with a set of observed covariates $\left(
X,W,Z\right) $. Assume that the $j$th failure time satisfies
\begin{equation}
\ln T_{j}=a_j(W,X)-\varepsilon_j, \label{competing risk model}
\end{equation}for some function $a_{j}(w,x)$ , $j=1,...,J$. The model may then be termed a multivariate generalized accelerated failure time model \citep[][]{Kalbfleisch2002,Fosgerau2013y}. The econometrician has knowledge of
\begin{equation}
 \Pi _{j}\left( w,x,z\right) := E\left[ \ln Y| W=w,X=x,Z=z\right]\cdot P\left( D=j|W=w,X=x,Z=z \right), \label{eq: Pi competing risk}
\end{equation} for $j=1,...,J$, where $Z$  is used to control for potential dependence between $(W,X)$ and $\varepsilon$. We collect the unobservables in $\varepsilon =\left( \varepsilon _{1},...,\varepsilon
_{J}\right) $ and again require them to be conditionally independent of $(X,W)$ in which case, as shown in Section \ref{Sec:Application}, $\Pi$ defined above again satisfies eq. (\ref{eq:model}).

Typical applications of the above model are in the modelling of (un)employment spells where an exit from the unemployment register can be the result of finding a full or a part-time job in different sectors or another change of status. Thus, in this setting, $j=1,...J$  indices the different exits (types of non-unemployment), and $(W,X)$ contain both variables characterizing the types of employment (such as salary in a given type/sector of employment) and individual-specific  controls (such as age and marital status). Similar to discrete choice models, we would then need to construct $g(W)$ to capture risk-specific characteristics with continuous variation and then include all other co-variates in $X$. Most empirical applications assume a parametric structure for the index, e.g. $a(W,X) = \alpha W + \beta X$. In this setting, requiring $g$ to be known effectively assumes fixing $\alpha\in \mathbb{R}^{J \times d_Y}$ . At the same time, we impose very weak restrictions on the distributional features of the regressors $X$ and how they enter the index $a(W,X)$.

\section{General framework}

\label{Sec: Identification}

We now return to the general model given in eqs. (\ref{eq:model}-\ref{eq: a(w,x) def}) where $g:\mathbb{R}^{d_Y}\rightarrow \mathbb{R}^{J}$ is assumed to be a known function while $h:%
\mathbb{R}^{d_{X}}\rightarrow \mathbb{R}^{J}$ and $\Lambda :\mathbb{R}%
^{J}\times \mathbb{R}^{d_{Z}}\rightarrow \mathbb{R}^{J}$ are unknown functions.  In the following, let $\mathrm{int}\mathcal{A}$ denote the interior of a given set $%
\mathcal{A}$ and let $\mathrm{supp}\left(Y\right) $ denote the support of a given random variable $Y$. We then take $\Pi \left( w,x,z\right) $ as given and known to us for all $\left( w,x,z\right) \in \mathrm{supp}\left( W,X,Z\right) \subseteq \mathbb{R}^{J}\times \mathbb{R}^{d_{X}}\times \mathbb{R}^{d_{Z}}$ where $%
\left( W,X,Z\right) $ denote the random variables that we have observed, c.f. the examples in the previous section.

The covariates contained in $g \left(W \right)$ play a special role in our approach in that we need sufficient continuous variation in these to achieve identification. First note that $\dim g\left( W\right) =J$. Thus, sufficient continuous variation of $g(W)$, which is known to us, permit us to identify the relative variation of $\Lambda \left( a,z\right) $ w.r.t. $a$. Formally, for any given pair $\left(x,z \right)\in \mathrm{supp}\left( X,Z\right) $, define 
\begin{equation}
\mathcal{G}\left( x,z\right) =\mathrm{int\ supp}\left( g\left( W\right)
|X=x,Z=z\right) ,\text{ \ \ }\mathcal{X}\left( z\right) =\mathrm{supp}\left(
X|Z=z\right)   \label{eq: G(x,z) def}
\end{equation}We will then throughout implicitly require that some of the open sets $\mathcal{G}\left( x,z\right)$,  $\left(x,z \right)\in \mathrm{supp}\left( X,Z\right) $, are non-empty and then achieve identification at the values of $x$ for which this is true.
A sufficient condition for a given  $\mathcal{G}(x,z)$ to be non-empty is that the distribution of $W|X=x,Z=z$ is continuous and that $g$ maps open sets into open sets; however, this is not required and $g(W)$ may contain discrete components as long as they fall within the support of the continuous component. However, our identification result still applies if any values of a discrete component fall outside the continuous support but excludes these values. This also rules out that some components of $W$ are included in $Z$ since in this case $\mathcal{G}\left( x,z\right) =\emptyset $. At the same time, however, $\left( X,W\right) $ can depend on $Z$; we just need sufficient variation in $\left( X,W\right) $ conditional on $Z$. Moreover, no continuity restrictions are imposed on the distribution of  $\left( X,Z\right) $ which may be completely discrete.
Finally, we would like to stress that we do not impose any large-support restrictions on $g(W)$, which is in contrast to most existing results in the literature, as discussed  in the Introduction. If, for example, $\mathcal{G}\left( x,z\right) = \mathbb{R}^J$, for all $x$, then our result demonstrates that $h(x)$ is identified on all of  $\mathrm{supp}\left( X\right)$ ; but it is not necessary, identification on all of  $\mathrm{supp}\left( X\right)$ can be achieved without such full support condition.

Next, let
\[
\mathcal{M}\left( x,z\right) =\mathcal{G}\left( x,z\right) \times \left\{x\right\} ,\text{ \ \ } \mathcal{A}\left( x,z\right) =a\left( \mathcal{M}\left( x,z\right) \right) =\mathcal{G}\left( x,z\right) +\left\{ h\left( x\right) \right\}  ,
\]
denote the support of $\left(W,X\right)|\left(X,Z\right)=\left(x,z\right)$ and $a\left(W,X\right)|\left(X,Z\right)=\left(x,z\right)$, respectively, and
\begin{equation}
\mathcal{M}\left( z\right) =\mathcal{\cup }_{x\in \mathcal{X}%
\left( z\right) }\mathcal{M}\left( x,z\right) , \text{ \ \ }
\mathcal{A}\left( z\right) =\mathcal{\cup }_{x\in \mathcal{X}\left( z\right)
}\mathcal{A}\left( x,z\right) =a\left( \mathcal{M}\left( z\right) \right) ,\label{eq: A(z) def}
\end{equation}
the supports of the same random variables but now only conditioning on $Z=z$. Finally, for some set  $\mathcal{Z}_{0} \subseteq\mathrm{supp}\left( Z\right)$  chosen according to certain assumptions stated below, let
\begin{equation}
\mathcal{A}_0 =\cup _{z\in \mathcal{Z}_0} \mathcal{A}\left( z\right)
\text{ \ \ }\mathcal{X}_0 =\cup_{z\in \mathcal{Z}_0} \mathcal{X} \left( z\right) .  \label{eq: A_0 def}
\end{equation}be the supports of $a\left(W,X\right)$ and $X$ conditional on  $Z\in \mathcal{Z}_0$, respectively. We will then show identification of  $h\left(x\right) $ and $\Lambda \left( a,z\right) $ for $x \in
\mathcal{X}_{0}$, $a\in \mathcal{A}_{0}$ and $z\in \mathcal{Z}_{0}$. Specifically, $\mathcal{Z}_{0}$ will be constructed according to certain properties of the underlying covariates and the functions of interest. Observe the dependence of $\mathcal{M}_{0}$ and $\mathcal{A}_{0}$ on the set $\mathcal{Z}_{0} $. To achieve \textquotedblleft maximal\textquotedblright\ identification, we would ideally like to choose $\mathcal{Z}_{0}=\mathrm{supp}\left( Z\right) $. However, we potentially have to restrict $\mathcal{Z}_{0}$. First, we require $a\mapsto \Lambda \left( a,z\right) $ to satisfy the following condition for all $z \in \mathcal{Z}_0$:

\begin{assumption} \label{A:injective}For any $z\in\mathcal{Z}_{0}$, $a\mapsto\Lambda\left(a,z%
\right)$ is injective on $\mathcal{A}\left(z\right)$ as defined in (\ref{eq: A(z) def}).
\end{assumption}

By asking for $\Lambda (a,z) $ to be injective, we can identify the relative variation in $a \left(w,x\right)$ through the observed variation in $\Pi \left(w,x,z \right)$. In a given application, Assumption \ref{A:injective} may not hold for all $z\in \mathrm{supp}\left( Z\right) $ in which case we need to remove such values from $\mathcal{Z}_{0}$. In the worst case scenario, this leaves us with $\mathcal{Z}_{0}$ being empty and our identification result becomes void. At the other extreme, $\mathcal{Z}_{0}=\mathrm{supp}\left(
Z\right) $ and we may achieve identification on the whole support. 

Due to the structure of $a\left( w,x\right) $, it follows from the definition of $\mathcal{G}\left(x, z\right) $ that $\mathcal{A}\left(x, z\right) $ and thereby also $ \mathcal{A} \left(z\right)$ and $\mathcal{A}_{0}$ are open sets. We add to this by also requiring $ \mathcal{A} \left(z\right)$ to be connected  for all $z \in \mathcal{Z}_0$. An open set $\mathcal{A}$ is connected if  $\mathcal{A}= \mathcal{O}_{1}\cup \mathcal{O}_{2}$ implies that $ \mathcal{O}_{1}\cap \mathcal{O}%
_{2}\neq \varnothing$  whenever  $\mathcal{O}_{1}$ and $\mathcal{O}_{2}$ are nonempty open sets. Thus an open connected set cannot be separated into two non-empty disjoint open sets. We then impose:

\begin{assumption}\label{A:con1} $\mathcal{A}(z)$  is connected for all $z\in \mathcal{Z}_0$.
\end{assumption}

Assumption \ref{A:con1} allows us to go from local identification at a given point $x \in  \mathcal{X}(z)$ to relative identification on all of $\mathcal{X}(z)$, $z\in \mathcal{Z}_0$ via the image of  $a\left(x,w\right)$. The assumption imposes restrictions on the support of the  random variable $a\left(X,W\right)$ instead of $\left(X,W\right)$ themselves. This is done in order to impose minimal restrictions on the distribution of $X$ and the smoothness of $h$. Recall that $W$ is assumed to contain a continuous component. Thus, Assumption \ref{A:con1}  includes, for example, the case of $X$  being unbounded and discrete, or $X$ to be continuous while $h\left(X\right)$ is discontinuous everywhere.
Assumption \ref{A:con1} is not verifiable from data but the same holds for smoothness conditions that are regularly imposed in existing identification results. If we are willing to entertain certain smoothness conditions, such as the inverse of $\Lambda(a,z)$  being continuous with respect to $a$, then the assumption is implied by connectedness of $\Pi(\mathcal{M}(z)|z)= \Lambda(\mathcal{A}(z),z)$, this latter property being verifiable. Similarly, if we restrict $X$ and $h$ to both be continuous, it will be implied by connectedness of $\mathcal{M}\left( z\right) $.

Once we have achieved relative identification on each $\mathcal{X}(z)$, $z\in \mathcal{Z}_0$, global identification is then reached through the following assumption: 

\begin{assumption}\label{A:con2} If $\mathcal{Z}_1 \cup \mathcal{Z}_2 =\mathcal{Z}_0, \mathcal{Z}_1 , \mathcal{Z}_2 \neq \emptyset$, then $(\cup_{z\in\mathcal{Z}_1}\mathcal{M}(z) \cap (\cup_{z\in\mathcal{Z}_2}\mathcal{M}(z) \neq \emptyset$.
\end{assumption} 

This is used to paste together the relatively identified sets  $\mathcal{X}(z)$ across $z$. Again, this assumption does not require $X$ and/or $h$ to be continuous, only that the sets $\mathrm{supp}\left( W,X|Z=z\right) $, $z \in \mathcal{Z}_0$ overlap. Finally, the following normalization on the function $h$ gives us identification on $\mathcal{X}(z_0)$:

\begin{assumption} \label{A:normalisation}There exists known $z_{0}\in\mathcal{Z}_{0}$ and $%
\left(w_{0},x_{0}\right)\in\mathcal{M}\left(z_{0}\right)$ so that $%
h\left(x_{0}\right)=0$.
\end{assumption}

Such a normalization is needed to identify the level of $h$ since, for any given pair of $%
\left( \Lambda ,h\right) $, we have $\Lambda \left( g\left( w\right)
+h\left( x\right) ,z\right) =\tilde{\Lambda}\left( g\left( w\right) +\tilde{h%
}\left( x\right) ,z\right) $ where $\tilde{\Lambda}\left( a,z\right)
=\Lambda \left( a+c,z\right) $ and $\tilde{h}\left( x\right) =h\left(
x\right) -c$ for some given value of $c\in \mathbb{R}^{J}$.

\section{Main result}

\label{sec:results}

As explained earlier, we shall make use of the  notion of relative identification in our proof of identification. As a first step, we show relative identification on any two overlapping images of $a$; this is achieved through injectivity of $\Lambda\left(a,z\right)$ which allows us to map the overlapping images into overlapping images of $\Pi$.

\begin{lemma}
\label{Lem: A overlap}Suppose that Assumption \ref{A:injective} holds, and that $h\left( x^{\ast}\right) $ is identified at $x^{\ast}\in \mathcal{X}\left( z\right) $ for some $z\in \mathcal{Z}_{0}$. Then the set $\mathcal{X}^{\ast}\left(z\right) := \left\{x \in \mathcal{X}(z)|\mathcal{A}\left( x^{\ast},z\right) \cap \mathcal{A}\left( x,z\right) \neq \emptyset \right\} $ is identified and $h\left( x\right) $ is identified on $\mathcal{X}^{\ast}\left(z\right)$.
\end{lemma}

\begin{proof} By definition, $\mathcal{A}\left( x^{\ast},z\right) \cap \mathcal{A}\left( x,z\right) \neq \emptyset$ if and only if there exists $w^{\ast}$ and $w$ so that $g\left(w^{\ast}\right) \in \mathcal{G}\left( x^{\ast},z\right) $, $g\left(w\right) \in  \mathcal{G}\left( x,z\right)$ and $a\left(  w^{\ast} ,x^{\ast}\right) = a\left(  w ,x\right)$. Using that $\Lambda \left( a|z\right) $ is injective by Assumption \ref{A:injective}, the last equality is equivalent to $\Lambda \left( a\left(  w^{\ast} ,x^{\ast}\right),z\right) =\Lambda \left( a\left(  w ,x\right),z\right)$, which we recognize as
\begin{equation}
\Pi \left( w^{\ast},x^{\ast},z\right) =\Pi \left( w,x,z\right) ,  \label{eq: Pi solution}
\end{equation}
where $\Pi$ is known to us. Thus, $\mathcal{X}^{\ast}\left(x, z\right)$ is identified as the set of solutions $x$  to (\ref{eq: Pi solution}) as we vary $(w^{\ast},w)$. Next, for any given $x \in \mathcal{X}^{\ast}\left(x, z\right)$, let $w^{\ast}$ and $w$ be the corresponding values for which (\ref{eq: Pi solution})  holds. Since these are known, the value $
a\left( w^{\ast},x^{\ast}\right) =g\left( w^{\ast}\right) +h\left( x^{\ast}\right) $ is also known to us. This in turn implies that $h\left( x\right) =a\left( w,x\right) -g\left( w\right)
=a\left( w^{\ast},x^{\ast}\right) -g\left( w\right) $ is identified.
\end{proof}

We then use this lemma in conjunction with the connectedness of $\mathcal{A}(z)$ to show relative identification on each of the sets $\mathcal{X}\left( z\right) $:

\begin{lemma}
\label{Lem: X(z) ID}Suppose that Assumptions \ref{A:injective}-\ref{A:con1} hold. Then, for all $z\in \mathcal{Z}_{0}$, $h\left( x\right) $ is relatively identified on $\mathcal{X}\left( z\right) $ as defined in eq. (\ref{eq: G(x,z) def}).
\end{lemma}

\begin{proof}
Let  $x^{\ast }\in \mathcal{X}\left( z\right) $ be given and suppose we know the value of $h\left( x^{\ast }\right) $. Let  $\mathcal{X}^{\ast }\left( z\right) \subseteq \mathcal{X}\left( z\right)$ be the set on which $h\left( x\right) $ is identified and let $\mathcal{A}^{\ast }\left( z\right) =\cup _{x\in \mathcal{X}^{\ast }\left( z\right) } \mathcal{A}\left( x,z\right) $ be the corresponding values of $a\left( w,x\right) $. By assumption $x^{\ast }\in \mathcal{X}^{\ast }\left( z\right) $ and so the identified set is non-empty. This in turn implies that $\mathcal{A}^{\ast }\left( z\right) $ is non-empty and open. Now, seeking a contradiction, suppose that $\mathcal{X}^{\ast \ast }\left( z\right) :=\mathcal{X}\left( z\right) \backslash \mathcal{X}^{\ast }\left( z\right) \neq \varnothing $. Then define $\mathcal{A}^{\ast \ast }\left( z\right) =\cup _{x\in \mathcal{X}^{\ast \ast }\left( z\right) }\mathcal{A}\left(x,z\right) $ which is also open and non-empty. Since $\mathcal{A}^{\ast}\left( z\right) \cup \mathcal{A}^{\ast \ast }\left( z\right) =\mathcal{A}\left( z\right) $, which is connected according to Assumption \ref{A:con1}, there must exist $x\in \mathcal{X}^{\ast }\left( z\right) $ and $x^{\prime }\in \mathcal{X}^{\ast \ast }\left(z\right) $ so that $\mathcal{A}\left( x,z\right) \cap \mathcal{A}\left(x^{\prime },z\right) \neq \varnothing $. Lemma \ref{Lem: A overlap} then implies that $x^{\prime}$ and $h\left( x^{\prime }\right) $ is also identified which is a contradiction.
\end{proof}

Finally, the "connectedness" of $\cup_{z\in\mathcal{Z}_0}\mathcal{M}(z)$ as stated in Assumption \ref{A:con2}  together with the normalization in Assumption \ref{A:normalisation} gives us global identification:

\begin{theorem}
Under Assumptions \ref{A:injective}-\ref{A:normalisation}, $h\left( x\right) $ is identified on $\mathcal{X}_{0}=\cup _{z\in \mathcal{Z}%
_{0}}\mathcal{X}\left( z\right) $.
\end{theorem}

\begin{proof}
Let $\mathcal{X}^{\ast }$ be the identified set. By Lemma \ref{Lem: X(z) ID}, $\mathcal{X}^{\ast }=\cup _{z\in \mathcal{Z}^{\ast }}\mathcal{X}\left(z\right) $ for some $\mathcal{Z}^{\ast }\subseteq \mathcal{Z}_{0}$. By Assumption 4, $z_{0}\in \mathcal{Z}^{\ast }$ and so the set is non-empty. Seeking a contradiction, suppose that $\mathcal{Z}^{\ast \ast }:=\mathcal{Z}%
_{0}\backslash \mathcal{Z}^{\ast }\neq \varnothing $. By definition $%
\mathcal{Z}^{\ast }\cup \mathcal{Z}^{\ast \ast }=\mathcal{Z}_{0}$ and so $%
\left\{ \cup _{z\in \mathcal{Z}^{\ast }}\mathcal{M}(z)\right\} \cap \left\{
\cup _{z\in \mathcal{Z}^{\ast \ast }}\mathcal{M}(z)\right\} \neq \emptyset $
by Assumption \ref{A:con2}. This implies that there exists $z^{\ast }\in \mathcal{Z}%
^{\ast }$ and $z^{\ast \ast }\in \mathcal{Z}^{\ast \ast }$ so that $\mathcal{%
M}(z^{\ast })\cap \mathcal{M}(z^{\ast \ast })\neq \emptyset $ which in turn implies that there exists $x^{\ast }\in \mathcal{X}\left( z^{\ast }\right)
\cap \mathcal{X}\left( z^{\ast \ast }\right) $ for which $h\left( x^{\ast
}\right) $ is identified. But then Lemma \ref{Lem: X(z) ID} implies that $%
h\left( x\right) $ is identified on all of $\mathcal{X}\left( z^{\ast \ast
}\right) $ which is a contradiction.
\end{proof}
 
Once we have identified $h$ we can also identify $\Lambda$:

\begin{theorem}
\label{Th: Lamda id}Under Assumptions \ref{A:injective}-\ref{A:normalisation}, $\Lambda \left(a,z \right)$ is identified on  $\left\{\left(a,z \right)| a \in\mathcal{A}(z), z \in\mathcal{Z}_0 \right\}$.
\end{theorem}

\begin{proof}
Let $z\in\mathcal{Z}_{0}$ and $a\in\mathcal{A}\left(z\right)$ be given. By definition of $\mathcal{A}\left(z\right)$, there exists some pair $%
\left(w,x\right)\in\mathcal{M}\left(z\right)$ such that $a=a\left(w,x\right)$. Since $h(\cdot)$ and thereby also $a\left(\cdot,\cdot\right)$ is identified, the pair $%
\left(w,x\right) $ is known. But then we also know $%
\Pi\left(w,x,z\right)$ and so $\Lambda\left(a,z\right)=\Pi\left(w,x,z\right)$ is uniquely identified.
\end{proof}

\section{Applications}\label{Sec:Application}

This section applies the general result to the two main examples of Section \ref{Sec: Examples}, the ARUM and the competing risk model, and compare our identification results for these two models with existing ones found in the literature.  In both examples, we impose the following conditional independence restriction on the error term:

\begin{assumption} \label{A:cond independence}(i) $\varepsilon$ is conditionally independent of $\left(X,W\right) $, $F_{\varepsilon |\left( W,X,Z\right)
}\left(\cdot|\cdot,\cdot,z\right)=F_{\varepsilon |Z}\left(\cdot|z\right)$ for all $z \in \mathcal{Z}_0$ for some $\mathcal{Z}_0 \subseteq \mathrm{supp}\left(Z\right) $; (ii) $F_{\varepsilon |Z}\left(\cdot|z\right)$  has a conditional density with full support for all $z \in \mathcal{Z}_0$.
\end{assumption}
We demonstrate in the next two subsections that part (i) implies $\Pi$, as defined in eq. (\ref{eq: Pi ARUM}) and \ref{eq: Pi competing risk}, respectively, can be written on the form (\ref{eq:model})-(\ref{eq: a(w,x) def}) for all $z \in \mathcal{Z}_0$, while part (ii) ensures that the model specific $\Lambda(a,z)$ is injective w.r.t $a$ for all $z \in \mathcal{Z}_0$.

\subsection{ARUM}
Define the surplus function
\begin{equation*}
G\left( a_{0},...a_{J},z\right) :=E\left[ \max_{j=0,...,J}U_j|a\left( W,X\right)  =a,Z=z%
\right] =E\left[ \max_{j=0,...,J}\left\{ \varepsilon
_{j}+a_{j}\right\} |Z=z\right] ,
\end{equation*}for any given $\left( a_{0},a_{1},...,a_{J}\right) \in \mathbb{R}^{J+1}$, where the second equality uses eq. (\ref{ARUM model}) and Assumption \ref{A:cond independence}(i). The Williams-Daly-Zacchary Theorem \citep[][]{McFadden1981} then implies that the CCP's, as defined in (\ref{eq: Pi ARUM}), can be written on the form (\ref{eq:model})-(\ref{eq: a(w,x) def}) with $\Lambda$  defined as the gradient of the surplus function, 
\begin{equation*}
\Lambda \left( a,z\right): =\left. \frac{%
\partial G\left( a,z\right) }{\partial a}\right\vert
_{a_{0}=0}.
\end{equation*}We conclude:
\begin{corollary}
Any ARUM on the form (\ref{ARUM model}) that satisfies Assumptions \ref{A:injective}-\ref{A:normalisation} and \ref{A:cond independence}(i) is identified. 
\end{corollary}

Next, we discuss each of Assumptions \ref{A:injective}-\ref{A:normalisation} in the context of ARUM and how these compare with existing ones found in the literature on identification of ARUM.

First, Assumption \ref{A:injective}, injectivity of $\Lambda \left( \cdot ,z\right) $ for each $z$, is implied by Assumption \ref{A:cond independence}(ii), c.f. \citet[Thm 2.1]{Hofbauer2002}. However, Assumption \ref{A:cond independence}(ii) is not necessary  for injectivity to hold. A simply example is the binomial model, where the probability for alternative 0 is the cumulative distribution of $\varepsilon _{1}$. If the distribution includes point masses, then ties can occur, but this does not destroy injectivity. This is true for any tie-breaking rule. More generally, if the subdifferential of the surplus function is strictly cyclically monotone \citep[][]{Rockafellar1970}, which does not require the existence of a density, then the utility maximizing choice probabilities under any tie-breaking rule are injective \citep[][]{Sorensen2019}.

Assumptions \ref{A:con1}-\ref{A:con2} impose restrictions on the joint variation of $\left(g(W),X\right)$. For Assumption \ref{A:con1} to hold, we need to identify $J$ regressors, $g(W)$, that exhibit enough joint continuous variation so their joint support, conditional on  $(X,Z)$ has non-empty interior on $\mathbb{R}^J$. One instance where this can be achieved is if we have observed alternative specific characteristics. In case of demand modellling, one such choice would be a (transformation) of the (relative) prices of the different alternative while $X$ contains all remaining regressors, possibly including other alternative specific covariates. In this case, to control for potential endogeneity of prices, we could then include cost shifters in $Z$. Prices tend to exhibit continuous variation and Assumptions \ref{A:con1} would be likely to hold. Assumption \ref{A:con2} requires other observed product characteristics and the agent's observed characteristics to exhibit sufficient variation conditional on the controls in $Z$ so that these have overlapping support across different values of $Z$.

As already mentioned in the introduction, there are few fully nonparametric identification results for ARUM. To our knowledge, the only results comparable to ours are found in \cite{Matzkin1993}. Her results also require the presence of alternative specific regressors but impose stronger conditions on these and other covariates. Moreover, her set-up does not include any control variables. On the other hand, she does not necessarily require that $a(W,X)$ is additive, which we assume throughout. Theorem 1 of \cite{Matzkin1993} does allow for dependence between $(W,X)$  and $\varepsilon$ but in this case, she requires the observed component of the utilities to be identical across alternatives and strictly increasing in one of the arguments. In our notation, this requires $a_j(W,X)$, $j=1,...,J$ to all be identical. We do not impose any such constraints. Her Theorem 2 requires full independence between  $(W,X)$  and $\varepsilon$ but, on the other hand, impose fewer restrictions on  $a(W,X)$ compared to us. But in both cases, she identifies $\Lambda$  by letting different components of $W$ diverge to $+\infty$, which is an example of "thin set identification" discussed earlier.

\subsection{Perturbed discrete choice}
We here demonstrate that the CCP's for the pertubed discrete choice model again can be expressed on the form  (\ref{eq:model})-(\ref{eq: a(w,x) def})  with $\Lambda $ defined in (\ref{pertubed utility}) being injective. This is done under the following restrictions: First, in order to rule out zero demands, the norm of the gradient $\nabla _{q}\Omega \left( q|z\right) $ has to approach infinity as $q$ approaches the boundary of the unit simplex. Second, $\Omega \left( q|z\right) $ is differentiable\footnote{Note we do not require a Hessian.}. Third, we normalize the outside option so that $g_{0}\left( w\right) =h_{0}\left( x\right) =0$. Under these three restrictions, for each value of the control $z$, the demand solves the first-order condition for an interior solution,
\begin{equation*}
a+\nabla _{q}\Omega \left( \Lambda \left( a,z\right) |z\right) =\lambda
\iota ,
\end{equation*}where $\lambda $ is a scalar constant and $\iota \in \mathbb{R}^{J}$ is a vector consisting of ones. To show that $\Lambda $ is injective, consider this equation at $a_{1}$ and $a_{2}$ and assume that $\Lambda \left(
a_{1},z\right) =\Lambda \left( a_{2},z\right) $. Define a matrix $M$ such that $Mx=x-x_{0}\iota $ for all $x=\left( x_{0},...,x_{J}\right) \in \mathbb{R}^{J+1}$. Pre-multiply this matrix onto the first-order condition to obtain that
\begin{equation*}
a_{1}+M\nabla _{q}\Omega \left( \Lambda \left( a_{1},z\right) |z\right)
=a_{2}+M\nabla _{q}\Omega \left( \Lambda \left( a_{2},z\right) |z\right) ,
\end{equation*}which implies that $a_{1}=a_{2}$ as required. 

\subsection{Competing Risk}
Define
\begin{equation}
\Lambda \left( a,z\right) :=G\left( a,z\right) \cdot \frac{\partial
G\left( a,z\right) }{\partial a}, \label{eq: Lambda competing risk}
\end{equation}where as before $a=\left( a_{1},...,a_{J}\right) $ while $G\left( a,z\right)
$ is now defined as the expected log failure time,
\begin{equation*}
G\left( a,z\right) :=E\left[ \ln Y|a\left( W,X\right)  =a,Z=z%
\right] =-E\left[ \max_{j=1,...,J}\left\{ -a_{j}+\varepsilon _{j}\right\}
|Z=z\right] ,
\end{equation*}where the second equality uses  eq. (\ref{competing risk model}) and Assumption \ref{A:cond independence}(i). Williams-Daly-Zacchary Theorem \citep[][]{McFadden1981}  then implies that $\Pi$, now defined by (\ref{eq: Pi competing risk}), can be written on the form (\ref{eq:model})-(\ref{eq: a(w,x) def}). Injectivity of $\Lambda \left( a,z\right)  $, as given in eq. (\ref{eq: Lambda competing risk}), is obtained by recycling the arguments of the previous subsection except that no normalization of one of the causes of failure is required since the level $G\left( a,z\right) $ is included.
\begin{corollary}
Any competing risk model on the form (\ref{competing risk model}) that satisfies Assumptions \ref{A:injective}-\ref{A:normalisation} and \ref{A:cond independence}(i) is identified. 
\end{corollary}
Given that the competing risk model and the ARUM share a similar structure, the discussion of the remaining assumptions carry over to the current setting with obvious modifications.

Compared to existing results \citep[][]{Heckman1989,Lee2013}  we impose stronger conditions on the index $a(W,X)$ since we require it to be additive and with $g(W)$ known. On the other hand, \cite{Heckman1989} require $a(W,X)$  to go to zero as $W$ diverges, and so relies on a "thin set identification" argument, while \cite{Lee2013}  rely on a high-level functional rank-condition. It is unclear which primitive conditions suffice for this rank condition to hold. Finally, \citet{Honore2006}  restrict themselves to the case of purely discrete regressors and are only able to derive bounds for objects of interest. We achieve point identification as long as there is some continuous variation in $W$ while $X$ can be completely discrete

\section{Conclusion}\label{Sec:Conclusion}

We have established an identification result for a wide class of index models based on general topological arguments. Three key features of our
argument is that smoothness of the model is not required; no large support condition is imposed on the regressors; and control variables may contribute to achieving identification. We leave the development of nonparametric estimators of the identified components for future research.

\bibliographystyle{chicago}
%\bibliography{c:/mfbib/library}
\bibliography{references}

\end{document}